\documentclass[submission,copyright]{eptcs}
\providecommand{\event}{GandALF 2023} 

\usepackage{iftex}

\ifpdf
  \usepackage{underscore}         
  \usepackage[T1]{fontenc}        
\else
  \usepackage{breakurl}           
\fi

\title{TSO Games \\ \Large{On the decidability of safety games under the total store order semantics}}
\author{
    Stephan Spengler
    \institute{Uppsala University \\ Uppsala, Sweden}
    \email{stephan.spengler@it.uu.se}
    \and
    Sanchari Sil
    \institute{Chennai Mathematical Institute \\ Chennai, India}
    \email{sanchari@cmi.ac.in}
}
\def\titlerunning{TSO Games}
\def\authorrunning{S. Spengler \& S. Sil}

\usepackage{aliascnt}
\usepackage{amsmath}
\usepackage{amssymb}
\usepackage{amsthm}
\usepackage[capitalise, noabbrev]{cleveref}
\usepackage{booktabs}
\usepackage{enumerate}
\usepackage{float}
\usepackage{hyperref}
\usepackage{lineno}
\usepackage{multirow}
\usepackage{subcaption}
\usepackage{tabularx}
\usepackage{textcomp}
\usepackage{tikz}
\usepackage{xcolor}

\DeclareMathAlphabet{\mathcal}{OMS}{cmsy}{m}{n}
\newcolumntype{C}[1]{>{\centering\arraybackslash}p{#1}}
\usetikzlibrary{decorations.pathreplacing,shapes}

\newtheorem{theorem}{Theorem}

\newaliascnt{lemma}{theorem}
\newtheorem{lemma}[lemma]{Lemma}
\aliascntresetthe{lemma}

\newaliascnt{corollary}{theorem}

\aliascntresetthe{corollary}

\newaliascnt{claim}{theorem}
\newtheorem{claim}[claim]{Claim}
\aliascntresetthe{claim}

\theoremstyle{definition} 
\newaliascnt{remark}{theorem}
\newtheorem{remark}[remark]{Remark}
\aliascntresetthe{remark}

\newcommand{\inference}[3]{\textbf{#1} & \frac{#2}{#3}}

\newcommand{\Nat}{\mathbb{N}}
\newcommand{\bigO}{\mathcal{O}}

\newcommand{\newsemantics}[2]{\newcommand{#1}{\mathsf{#2}}}
\newcommand{\newfunction}[2]{\newcommand{#1}{\mathop\mathrm{#2}}} 
\newcommand{\newcomponent}[2]{\newcommand{#1}{\mathsf{#2}}}
\newcommand{\newinstruction}[2]{\newcommand{#1}{\mathtt{#2}}}
\newcommand{\renewinstruction}[2]{\renewcommand{#1}{\mathtt{#2}}}

\newcommand{\sizeof}[1]{|#1|}

\newcommand{\of}[1]{(#1)}
\newcommand{\tuple}[1]{\langle#1\rangle}
\newcommand{\set}[1]{\{#1\}}

\newcommand{\Set}[1]{\left\{#1\right\}}

\newcommand{\cof}[1][]{\ifthenelse{\isempty{#1}}{}{\of{#1}}}
\renewcommand{\to}[1][]{\mathop{\xrightarrow{~#1~}}}
\renewcommand{\part}{\rightharpoonup}
\newcomponent{\word}{w}

\newcommand{\newclass}[2]{\newcommand{#1}{\textsc{#2}}}
\newclass{\exptime}{ExpTime}
\newclass{\etime}{ETime}
\newclass{\nexptime}{NexpTime}
\newclass{\class}{Class}
\newclass{\pspace}{PSpace}
\newclass{\expspace}{ExpSpace}
\newclass{\dtime}{DTime}
\newclass{\dspace}{DSpace}

\newcommand{\TS}{\mathcal{T}}
\newcommand{\kstar}{^{*}}

\newcomponent{\conf}{c}
\newcomponent{\confset}{C}
\newcomponent{\lbl}{label}
\newcomponent{\lblset}{L}
\newcomponent{\state}{q}
\newcomponent{\stateset}{Q}

\newfunction{\post}{Post}
\newfunction{\pre}{Pre}

\newsemantics{\final}{final}
\newsemantics{\target}{target}
\newsemantics{\all}{all}
\newsemantics{\true}{true}
\newsemantics{\false}{false}

\newcommand{\program}{\mathcal{P}}
\newcomponent{\process}{Proc}
\newcomponent{\tso}{TSO}
\newcomponent{\dtso}{DTSO}
\newcomponent{\transition}{\delta}

\newcomponent{\xvar}{x}
\newcomponent{\varset}{Vars}
\newcomponent{\dval}{d}
\newcomponent{\valset}{Dom}
\newcomponent{\msg}{m}
\newcomponent{\instr}{instr}
\newcomponent{\instrs}{Instrs}
\newcommand{\xd}{\xvar, \dval}
\newcomponent{\self}{self}
\newcomponent{\other}{other}

\newinstruction{\rd}{rd}
\renewinstruction{\wr}{wr}
\newinstruction{\nop}{skip}
\newinstruction{\mf}{mf}
\newinstruction{\up}{up} 
\newinstruction{\prop}{prop} 
\newinstruction{\del}{del} 

\newcommand{\indexset}{\mathcal{I}}
\newcommand{\statemap}{\mathcal{S}}
\newcommand{\buffermap}{\mathcal{B}}
\newcommand{\memorymap}{\mathcal{M}}
\newcommand{\pid}{\iota}

\newcomponent{\view}{v}
\newcomponent{\viewset}{V}
\newcommand{\valuemap}{\mathcal{V}}
\newcommand{\fencemap}{\mathcal{F}}

\newcommand{\game}{\mathcal{G}}
\newcomponent{\play}{P}

\newcommand{\channelsystem}{\mathcal{L}}
\newcomponent{\channelstate}{s}
\newcomponent{\channelstateset}{S}
\newcomponent{\channelset}{L}
\newcomponent{\channelmessage}{m}
\newcomponent{\channelmessageset}{M}
\newcomponent{\channeloperation}{op}
\newcomponent{\channeloperationset}{Op}

\newcomponent{\letter}{\sigma}
\newcomponent{\alphabet}{\Sigma}
\newcomponent{\atmL}{L}
\newcomponent{\atmR}{R}
\newcomponent{\atmD}{D}
\newcomponent{\pos}{i}
\newcomponent{\posj}{j}

\newcommand{\xrd}{\xvar_\rd}
\newcommand{\xwr}{\xvar_\wr}
\newcomponent{\yvar}{y}
\newcomponent{\zvar}{z}

\newcomponent{\hstate}{h}
\newcomponent{\rstate}{r}

\begin{document}

\maketitle

\begin{abstract}
We consider an extension of the classical Total Store Order (TSO) semantics by expanding it to turn-based 2-player safety games.
During her turn, a player can select any of the communicating processes and perform its next transition.
We consider different formulations of the safety game problem depending on whether one player or both of them transfer messages from the process buffers to the shared memory.
We give the complete decidability picture for all the possible alternatives.
\end{abstract}

\section{Introduction}
Most modern architectures, such as Intel x86 \cite{x86-swdmanual-1-3}, SPARC \cite{sparc9}, IBM's POWER \cite{power-isa-v31b}, and ARM \cite{arm-v7ar-refman}, implement several relaxations and optimisations that reduce the latency of memory accesses. This has the effect of breaking the Sequential Consistency (SC) assumption \cite{DBLP:journals/tc/Lamport79}. SC is the classical strong semantics for concurrent programs that interleaves the parallel executions of processes while maintaining the order in which instructions were issued. Programmers usually assume that the execution of programs follows the SC model. However, this is not true when we consider concurrent programs running on modern architectures. In fact, even simple programs such as mutual exclusion and producer-consumer protocols, that are correct under SC, may exhibit erroneous behaviors. This is mainly due to the relaxation of the execution order of the instructions. For instance, a standard relaxation is to allow the reordering of reads and writes of the same process if the reads have been issued after the writes and they concern different memory locations. This relaxation can be implemented using an unbounded perfect FIFO queue/buffer between each process and the memory. These buffers are used to store delayed writes. The corresponding model is called Total Store Ordering (TSO) and corresponds to the formalisation of SPARC and Intel x86 ~\cite{DBLP:conf/tphol/OwensSS09,DBLP:journals/cacm/SewellSONM10}.

In TSO, an unbounded buffer is associated with each process. When a process executes a write operation, this write is appended to the end of the buffer of that process. A pending write operation on the variable $x$ at the head of a buffer can be deleted in a non-deterministic manner. This updates the value of the shared variable $x$ in the memory. To perform a read operation on a variable $x$, the process first checks its buffer for a pending write operation on the variable $x$. If such a write exists, then the process reads the value written by the newest pending write operation on $x$. Otherwise, the process fetches the value of the variable $x$ from the memory. The verification of programs running under TSO is challenging due to the unboundedness of the buffers. In fact, the induced state space of a program under TSO maybe infinite even if the program itself is a finite-state system.

The reachability problem for programs under TSO checks whether a given program state is reachable during program execution. It is also called safety problem, in case the target state is considered to be a bad state. It has been shown decidable using different alternative semantics for TSO (e.g., \cite{DBLP:conf/popl/AtigBBM10,DBLP:conf/tacas/AbdullaACLR12,DBLP:journals/lmcs/AbdullaABN18}). Furthermore, it has been shown in \cite{DBLP:conf/popl/AtigBBM10} that lossy channel systems (see e.g., \cite{wsts2,wsts1,DBLP:conf/icalp/AbdullaJ94,DBLP:journals/ipl/Schnoebelen02}) can be simulated by programs running under TSO. This entails that the reachability problem for programs under TSO is non-primitive recursive and that the repeated reachability problem is undecidable. This is an immediate consequence of the fact that the reachability problem for lossy channel is non-primitive recursive \cite{DBLP:journals/ipl/Schnoebelen02} and that the repeated reachability problem is undecidable \cite{DBLP:conf/icalp/AbdullaJ94}. The termination problem for programs running under TSO has been shown to be decidable in \cite{DBLP:journals/siglog/Atig20} using the framework of well-structured transition systems \cite{wsts1,wsts2}.

The authors of \cite{DBLP:conf/esop/BouajjaniDM13,DBLP:conf/icalp/BouajjaniMM11} consider the robustness problem for programs running under TSO. This problem consists in checking whether, for any given TSO execution, there is an equivalent SC execution of the same program. Two executions are declared equivalent by the robustness criterion if they agree on (1) the order in which instructions are executed within the same process (i.e., program order), (2) the write instruction from which each read instruction fetches its value (i.e., read-from relation), and (3) the order in which write instruction on the same variable are committed to memory (i.e., store ordering). The problem of checking whether a program is robust has been shown to be \pspace-complete in \cite{DBLP:conf/esop/BouajjaniDM13}. A variant of the robustness problem which is called persistence, declares that two runs are equivalent if (1) they have the same program order and (2) all write instructions reach the memory in the same order. Checking the persistency of a program under TSO has been shown to be \pspace-complete in \cite{DBLP:conf/esop/AbdullaAP15}. Observe that the persistency and robustness problems are stronger than the safety problem (i.e., if a program is safe under SC and robust/persistent, then it is also safe under TSO).

Due to the non-determinism of the buffer updates, the buffers associated with each process under TSO appear to exhibit a lossy behaviour. Previously, games on lossy channel systems (and more general on monotonic systems) were studied in \cite{DBLP:journals/logcom/AbdullaBd08}. Unfortunately these results are not applicable / transferable to programs under TSO whose induced transition systems are not monotone \cite{DBLP:conf/popl/AtigBBM10}.

In this paper, we consider a natural continuation of the works on both the study of the decidability/complexity of the formal verification of programs under TSO and the study of games on concurrent systems. This is further motivated by the fact that formal games provide a framework to reason about a system's behaviour, which can be leveraged in control model checking, for example in controller synthesis problems.

In more detail, we consider (safety) games played on the transition systems induced by programs running under TSO. Given a program under TSO, we construct a game in which two players A and B take turns in executing instructions of the program. The goal of player B is to reach a given set of final configurations, while player A tries to avoid this. Thus, it can also be seen as a reachability game with respect to player B. In this game, the turn determines which player will execute the next program instruction. However, this definition leaves the control of updates undefined. To address this, we give the player the possibility to update memory by removing the pending writes from the buffer between the execution of two instructions.

The control over the buffer updates is shared between the two players in varying ways. We differentiate between multiple scenarios based on when exactly each player is allowed to update. In particular, for each player A or B we have the following cases: (1) she can never update, (2) she can update after her own turn, (3) she can update before her own turn, and (4) she can always update, i.e. before and after her own turn. In total, we obtain an exhaustive collection of 16 different TSO games. We divide these 16 games into four different groups, depending on their decidability results.
\begin{itemize}
    \item Group I (7 games) can be reduced to TSO games with 2-bounded buffers.
    \item Group II (1 game) can be reduced to TSO games with bounded buffers.
    \item Group III (7 games) can simulate perfect channel systems.
    \item Group IV (1 game) can be reduced to a finite game without buffers.
\end{itemize}
This classification is shown in \autoref{fig:tso-groups}. Of these four groups, only Group III is undecidable, the others each reduce to a finite game and are thus decidable.

\begin{figure}
\centering
\def\wdth{0.12\textwidth}
\begin{tabular}{p\wdth C\wdth C\wdth C\wdth C\wdth C\wdth}
    \toprule
    & & \multicolumn{4}{c}{Player A:} \\
    & & always & before & after & never \\
    \midrule
    \multirow{4}{0.15\textwidth}{Player B:}
    & always & I (d) & & \multicolumn{1}{c|}{} & \\
    \cline{4-4}
    & before & & \multicolumn{1}{|c|}{II (d)} & \multicolumn{1}{c|}{} & \\
    \cline{4-5}
    & after & \multicolumn{2}{c|}{} & III (u) & \\
    \cline{3-4}\cline{6-6}
    & never & \multicolumn{3}{c|}{} & IV (d) \\
    \bottomrule
\end{tabular}
\caption{Groups of TSO games, where players A and B are allowed to update the buffer: always, before their own move, after their own move, or never. The games in group I, II and IV are decidable (d), the games in group III are undecidable (u).}
\label{fig:tso-groups}
\end{figure}

Finally, we establish the exact computational complexity for the decidable games. In fact, we show that the problem is \exptime-complete. We prove \exptime-hardness by a reduction from the problem of acceptance of a word by a linearly bounded alternating Turing machine \cite{DBLP:journals/jacm/ChandraKS81}. To prove \exptime-membership, we show that it is possible to compute the set of winning region for player B in exponential time. These results are surprising given the non-primitive recursive complexity of the reachability problem for programs under TSO and the undecidability of the repeated reachability problem.

\smallskip

\noindent
{\bf Related Works.}
In addition to the related work mentioned in the introduction on the decidability / complexity of the verification problems of programs running under TSO, there have been some works on parameterized verification of programs running under TSO. The problem consists in verifying a concurrent program regardless of the number of involved processes (which are identical finite-state systems). The parameterised reachability problem of programs running under TSO has been shown to be decidable in \cite{DBLP:conf/concur/AbdullaABN16,DBLP:journals/lmcs/AbdullaABN18}. While this problem for concurrent programs performing only read and writing operations (no atomic read-write instructions) is \pspace-complete \cite{DBLP:journals/pacmpl/AbdullaAR20}. This result has been recently extended to processes manipulating abstract data types over infinite domains \cite{DBLP:conf/tacas/AbdullaAFGHKS23}. Checking the robustness of a parameterised concurrent system is decidable and \expspace-hard \cite{DBLP:conf/esop/BouajjaniDM13}.

As far as we know this is the first work that considers the game problem for programs running under TSO. The proofs and techniques used in this paper are different from the ones used to prove decidability / complexity results for the verification of programs under TSO except the undecidability result which uses some ideas from the reduction from the reachability problem for lossy channel systems to its corresponding problem for programs under TSO \cite{DBLP:conf/popl/AtigBBM10}. However, our undecidability proof requires us to implement a protocol that detects lossiness of messages in order to turn the lossy channel system into a perfect one (which is the most intricate part of the proof).

\section{Preliminaries}

\subsection{Transition Systems}
A \emph{(labeled) transition system} is a triple $\tuple{ \confset, \lblset, \to }$, where $\confset$ is a set of \emph{configurations}, $\lblset$ is a set of \emph{labels}, and $\to \subseteq \confset \times \lblset \times \confset$ is a \emph{transition relation}.
We usually write $\conf_1 \to[\lbl] \conf_2$ if $\tuple{ \conf_1, \lbl, \conf_2} \in \to$.
Furthermore, we write $\conf_1 \to \conf_2$ if there exists some $\lbl$ such that $\conf_1 \to[\lbl] \conf_2$.
A \emph{run} $\pi$ of $\TS$ is a sequence of transitions $\conf_0 \to[\lbl_1] \conf_1 \to[\lbl_2] \conf_2 \dots \to[\lbl_n] \conf_n$.
It is also written as $\conf_0 \to[\pi] \conf_n$.
A configuration $\conf'$ is \emph{reachable} from a configuration $\conf$, if there exists a run from $\conf$ to $\conf'$.

For a configuration $\conf$, we defined $\pre\of\conf = \set{ \conf' \mid \conf' \to \conf }$ and $\post\of\conf = \set{ \conf' \mid \conf \to \conf' }$.
We extend these notions to sets of configurations $\confset'$ with $\pre(\confset') = \bigcup_{\conf \in \confset'} \pre\of\conf$ and $\post(\confset') = \bigcup_{\conf \in \confset'} \post\of\conf$.

An \emph{unlabeled transition system} is a transition system without labels.
Formally, it is defined as a TS with a singleton label set.
In this case, we omit the labels.

\subsection{Perfect Channel Systems}
Given a set of messages $\channelmessageset$, define the set of channel operations $\channeloperationset = \set{ !\channelmessage, ?\channelmessage \mid \channelmessage \in \channelmessageset} \cup \set\nop$.
A \emph{perfect channel system} (PCS) is a triple $\channelsystem = \tuple{ \channelstateset, \channelmessageset, \transition }$, where $\channelstateset$ is a set of states, $\channelmessageset$ is a set of messages, and $\transition \subseteq \channelstateset \times \channeloperationset \times \channelstateset$ is a transition relation.
We write $\channelstate_1 \to[\channeloperation] \channelstate_2$ if $\tuple{ \channelstate_1, \channeloperation, \channelstate_2 } \in \transition$.

Intuitively, a PCS models a finite state automaton that is augmented by a \emph{perfect} (i.e. non-lossy) FIFO buffer, called \emph{channel}.
During a \emph{send operation} $!\channelmessage$, the channel system appends $\channelmessage$ to the tail of the channel.
A transition $?\channelmessage$ is called \emph{receive operation}.
It is only enabled if the channel is not empty and $\channelmessage$ is its oldest message.
When the channel system performs this operation, it removes $\channelmessage$ from the head of the channel.
Lastly, a $\nop$ operation just changes the state, but does not modify the buffer.

The formal semantics of $\channelsystem$ are defined by a transition system $\TS_\channelsystem = \tuple{ \confset_\channelsystem, \lblset_\channelsystem, \to_\channelsystem }$, where $\confset_\channelsystem = \channelstateset \times \channelmessageset\kstar$, $\lblset_\channelsystem = \channeloperationset$ and the transition relation $\to_\channelsystem$ is the smallest relation given by:
\begin{itemize}
	\item If $\channelstate_1 \to[!\channelmessage] \channelstate_2$ and $\word \in \channelmessageset\kstar$, then $\tuple{ \channelstate_1, \word } \to[!\channelmessage]_\channelsystem \tuple{ \channelstate_2, \channelmessage \cdot \word }$.
	\item If $\channelstate_1 \to[?\channelmessage] \channelstate_2$ and $\word \in \channelmessageset\kstar$, then $\tuple{ \channelstate_1, \word \cdot \channelmessage } \to[?\channelmessage]_\channelsystem \tuple{ \channelstate_2, \word }$.
	\item If $\channelstate_1 \to[\nop] \channelstate_2$ and $\word \in \channelmessageset\kstar$, then $\tuple{ \channelstate_1, \word } \to[\nop]_\channelsystem \tuple{ \channelstate_2, \word }$.
\end{itemize}
A state $\channelstate_F \in \channelstateset$ is \emph{reachable} from a configuration $\conf_0 \in \confset_\channelsystem$, if there exists a configuration $\conf_F = \tuple{ \channelstate_F, \word_F }$ such that $\conf_F$ is reachable from $\conf_0$ in $\TS_\channelsystem$.
The \textbf{state reachability problem} of PCS is, given a perfect channel system $\channelsystem$, an initial configuration $\conf_0 \in \confset_\channelsystem$ and a final state $\channelstate_F \in \channelstateset$, to decide whether $\channelstate_F$ is reachable from $\conf_0$ in $\TS_\channelsystem$.
It is undecidable \cite{DBLP:journals/jacm/BrandZ83}.

\section{Concurrent Programs}

\subsection{Syntax}

Let $\valset$ be a finite data domain and $\varset$ be a finite set of shared variables over $\valset$.
We define the \emph{instruction set} $\instrs = \set{ \rd\of\xd, \wr\of\xd \mid \xvar \in \varset, \dval \in \valset } \cup \set{ \nop, \mf }$,
which are called \emph{read}, \emph{write}, \emph{skip} and \emph{memory fence}, respectively.
A process is represented by a finite state labeled transition system.
It is given as the triple $\process = \tuple{ \stateset, \instrs, \transition }$, where $\stateset$ is a finite set of \emph{local states} and $\transition \subseteq \stateset \times \instrs \times \stateset$ is the transition relation.
As with transition systems, we write $\state_1 \to[\instr] \state_2$ if $\tuple{ \state_1, \instr, \state_2} \in \transition$ and $\state_1 \to \state_2$ if there exists some $\instr$ such that $\state_1 \to[\instr] \state_2$.

A \emph{concurrent program} is a tuple of processes $\program = \tuple{ \process^\pid }_{\pid \in \indexset}$, where $\indexset$ is a finite set of process identifiers.
For each $\pid \in \indexset$ we have $\process^\pid = \tuple{ \stateset^\pid, \instrs, \transition^\pid }$.
A \emph{global} state of $\program$ is a function $\statemap: \indexset \to \bigcup_{\pid \in \indexset} \stateset^\pid$ that maps each process to its local state, i.e $\statemap(\pid) \in \stateset^\pid$.

\subsection{TSO Semantics}

Under TSO semantics, the processes of a concurrent program do not interact with the shared memory directly, but indirectly through a FIFO \emph{store buffer} instead.
When performing a \emph{write} instruction $\wr\of\xd$, the process adds a new message $\tuple\xd$ to the tail of its store buffer.
A \emph{read} instruction $\rd\of\xd$ works differently depending on the current buffer content of the process.
If the buffer contains a write message on variable $\xvar$, the value $\dval$ must correspond to the value of the most recent such message.
Otherwise, the value is read directly from memory.
A \emph{skip} instruction only changes the local state of the process.
The \emph{memory fence} instruction is disabled, i.e. it cannot be executed, unless the buffer of the process is empty.
Additionally, at any point during the execution, the process can \emph{update} the write message at the head of its buffer to the memory.
For example, if the oldest message in the buffer is $\tuple\xd$, it will be removed from the buffer and the memory value of variable $\xvar$ will be updated to contain the value $\dval$.
This happens in a non-deterministic manner.

Formally, we introduce a TSO \emph{configuration} as a tuple $\conf = \tuple{ \statemap, \buffermap, \memorymap }$, where:
\begin{itemize}
	\item $\statemap: \indexset \to \bigcup_{\pid \in \indexset} \stateset^\pid$ is a global state of $\program$.
	\item $\buffermap: \indexset \to (\varset \times \valset)\kstar$ represents the buffer state of each process.
	\item $\memorymap: \varset \to \valset$ represents the memory state of each shared variable.
\end{itemize}
Given a configuration $\conf$, we write $\statemap\of\conf$, $\buffermap\of\conf$ and $\memorymap\of\conf$ for the global program state, buffer state and memory state of $\conf$.
The semantics of a concurrent program running under TSO is defined by a transition system $\TS_\program = \tuple{ \confset_\program, \lblset_\program, \to_\program }$,
where $\confset_\program$ is the set of all possible TSO configurations,
$\lblset_\program = \set{ \instr_\pid \mid \instr \in \instrs, \pid \in \indexset } \cup \set{ \up_\iota \mid \pid \in \indexset }$ is the set of labels.
The transition relation $\to_\program$ is given by the rules in \autoref{fig:tso-semantics}, where we use $\buffermap\of\pid|_{\set\xvar \times \valset}$ to denote the restriction of $\buffermap\of\pid$ to write messages on the variable $\xvar$.

\begin{figure}
\centering
\begin{equation*}
\begin{array}{lc}

\inference{read-own-write}
	{\state \to[\rd\of\xd] \state' \qquad \statemap\of\pid = \state \qquad \buffermap\of\pid|_{\set\xvar \times \valset} = \tuple\xd \cdot \word}
	{\tuple{ \statemap, \buffermap, \memorymap} \to[\rd\of\xd_\pid]_\program \tuple{ \statemap[\pid \leftarrow \state'], \buffermap, \memorymap}}
\bigskip\\
\inference{read-from-memory}
	{\state \to[\rd\of\xd] \state' \qquad \statemap\of\pid = \state \qquad \buffermap\of\pid|_{\set\xvar \times \valset} = \varepsilon \qquad \memorymap\of\xvar = \dval}
	{\tuple{ \statemap, \buffermap, \memorymap} \to[\rd\of\xd_\pid]_\program \tuple{ \statemap[\pid \leftarrow \state'], \buffermap, \memorymap}}
\bigskip\\
\inference{write}
	{\state \to[\wr\of\xd] \state' \qquad \statemap\of\pid = \state}
	{\tuple{ \statemap, \buffermap, \memorymap} \to[\wr\of\xd_\pid]_\program \tuple{ \statemap[\pid \leftarrow \state'], \buffermap[\pid \leftarrow \tuple\xd \cdot \buffermap\of\pid], \memorymap}}
\bigskip\\
\inference{skip}
	{\state \to[\nop] \state' \qquad \statemap\of\pid = \state}
	{\tuple{ \statemap, \buffermap, \memorymap} \to[\nop_\pid]_\program \tuple{ \statemap[\pid \leftarrow \state'], \buffermap, \memorymap}}
\bigskip\\
\inference{memory-fence}
	{\state \to[\mf] \state' \qquad \statemap\of\pid = \state \qquad \buffermap\of\pid = \varepsilon}
	{\tuple{ \statemap, \buffermap, \memorymap} \to[\mf_\pid]_\program \tuple{ \statemap[\pid \leftarrow \state'], \buffermap, \memorymap}}
\bigskip\\
\inference{update}
	{\buffermap\of\pid = \word \cdot \tuple\xd}
	{\tuple{ \statemap, \buffermap, \memorymap} \to[\up_\pid]_\program \tuple{ \statemap, \buffermap[\pid \leftarrow \word], \memorymap[\xvar \leftarrow \dval]}}

\end{array}
\end{equation*}
\caption{TSO semantics}
\label{fig:tso-semantics}
\end{figure}

A global state $\statemap_F$ is \emph{reachable} from an initial configuration $\conf_0$, if there is a configuration $\conf_F$ with $\statemap(\conf_F) = \statemap_F$ such that $\conf_F$ is reachable from $\conf_0$ in $\TS_\program$.
The \textbf{state reachability problem} of TSO is, given a program $\program$, an initial configuration $\conf_0$ and a final global state $\statemap_F$, to decide whether $\statemap_F$ is reachable from $\conf_0$ in $\TS_\program$.

We define $\up\kstar$ to be the transitive closure of $\set{ \up_\iota \mid \pid \in \indexset }$, i.e. $\conf_1 \to[\up\kstar]_\program \conf_2$ if and only if $\conf_2$ can be obtained from $\conf_1$ by some amount of buffer updates.

\section{Games}

\subsection{Definitions}

A \emph{(safety) game} is an unlabeled transition sytem, in which two players A and B take turns making a \emph{move} in the transition system, i.e. changing the state of the game from one configuration to an adjacent one.
The goal of player B is to reach a given set of final configurations, while player A tries to avoid this.
Thus, it can also be seen as a \emph{reachability} game with respect to player B.

Formally, a game is defined as a tuple $\game = \tuple{ \confset, \confset_A, \confset_B, \to, \confset_F}$, where $\confset$ is the set of configurations, $\confset_A$ and $\confset_B$ form a partition of $\confset$, the transition relation is restricted to $\to \subseteq (\confset_A \times \confset_B) \cup (\confset_B \times \confset_A)$, and $\confset_F \subseteq \confset_A$ is a set of \emph{final states}.
Furthermore, we assume without loss of generality that $\game$ is deadlock-free, i.e. $\post(\conf) \neq \emptyset$ for all $\conf \in \confset$.

A \emph{play} $\play$ of $\game$ is an infinite sequence $\conf_0, \conf_1, \dots$ such that $\conf_i \to \conf_{i+1}$ for all $i \in \Nat$.
In the context of safety games, $\play$ is \emph{winning} for player B if there is $i \in \Nat$ such that $\conf_i \in \confset_F$.
Otherwise, it is \emph{winning} for player A.
This means that player B tries to force the play into $\confset_F$, while player A tries to avoid this.

A \emph{strategy} of player A is a partial function $\sigma_A: \confset\kstar \part \confset_B$, such that $\sigma_A(\conf_0, \dots, \conf_n)$ is defined if and only if $\conf_0, \dots, \conf_n$ is a prefix of a play, $\conf_n \in \confset_A$ and $\sigma_A(\conf_0, \dots, \conf_n) \in \post(\conf_n)$.
A strategy $\sigma_A$ is called \emph{positional}, if it only depends on $\conf_n$, i.e. if $\sigma_A(\conf_0, \dots, \conf_n) = \sigma_A(\conf_n)$ for all $(\conf_0, \dots, \conf_n)$ on which $\sigma_A$ is defined.
Thus, a positional strategy is usually given as a total function $\sigma_A: \confset_A \to \confset_B$.
Given two games $\game$ and $\game'$ and a strategy $\sigma_A$ for $\game$, an \emph{extension} of $\sigma_A$ to $\game'$ is a strategy $\sigma_A'$ of $\game'$ that is also an extension of $\sigma_A$ to the configuration set of $\game'$ in the mathematical sense, i.e. $\sigma_A(\conf_0, \dots, \conf_n) = \sigma_A'(\conf_0, \dots, \conf_n)$ for all $(\conf_0, \dots, \conf_n)$ on which $\sigma_A$ is defined.
Conversely, $\sigma_A$ is called the \emph{restriction} of $\sigma_A'$ to $\game$.
For player B, strategies are defined accordingly.

Two strategies $\sigma_A$ and $\sigma_B$ together with an initial configuration $\conf_0$ induce a play $\play(\conf_0, \sigma_A, \sigma_B) = \conf_0, \conf_1, \dots$ such that $\conf_{i+1} = \sigma_A(\conf_0, \dots, \conf_i)$ for all $\conf_i \in \confset_A$ and $\conf_{i+1} = \sigma_B(\conf_0, \dots, \conf_i)$ for all $\conf_i \in \confset_B$.
A strategy $\sigma_A$ is \emph{winning} from a configuration $\conf$, if for \emph{all} strategies $\sigma_B$ it holds that $\play(\sigma_A, \sigma_B, \conf)$ is a winning play for player A.
A configuration $\conf$ is \emph{winning} for player A if she has a strategy that is winning from $\conf$.
Equivalent notions exist for player B.
The \textbf{safety problem} for a game $\game$ and a configuration $\conf$ is to decide whether $\conf$ is winning for player A.

\begin{lemma}[Proposition 2.21 in \cite{DBLP:conf/dagstuhl/Mazala01}]
\label{lem:positional}
    In safety games, every configuration is winning for exactly one player.
    A player with a winning strategy also has a positional winning strategy.
\end{lemma}

Since we only consider safety games in this paper, strategies will be considered to be positional unless explicitly stated otherwise.
Furthermore, \autoref{lem:positional} implies the following:
\begin{itemize}
    \item $\conf_A \in \confset_A$ is winning for player A $\iff$ there is $\conf_B \in \post(\conf_A)$ that is winning for player A.
    \item $\conf_B \in \confset_B$ is winning for player A $\iff$ all $\conf_A \in \post(\conf_B)$ are winning for player A.
\end{itemize}

A \emph{finite game} is a game with a finite set of configurations.
It is rather intuitive that the safety problem is decidable for finite games, e.g. by applying a backward induction algorithm.
In particular, the winning configurations for each player are computable in linear time:
\begin{lemma}[Chapter 2 in \cite{DBLP:conf/dagstuhl/2001automata}]
\label{lem:finite}
    Computing the set of winning configurations for a finite game with $n$ configurations and $m$ transitions is in $\bigO(n+m)$.
\end{lemma}

\subsection{TSO games}

A TSO program $\program = \tuple{\process^\pid}_{\pid \in \indexset}$ and a set of final local states $\stateset_F^\program \subseteq \stateset^\program$ induce a safety game $\game(\program,\stateset_F^\program) = \tuple{ \confset, \confset_A, \confset_B, \to, \confset_F }$ as follows.
The sets $\confset_A$ and $\confset_B$ are copies of the set $\confset^\program$ of TSO configurations, annotated by $A$ and $B$, respectively: $\confset_A := \set{ \conf_A \mid \conf \in \confset^\program}$ and $\confset_B := \set{ \conf_B \mid \conf \in \confset^\program}$.
The set of final configurations is defined as $\confset_F := \set{ \tuple{ \statemap, \buffermap, \memorymap }_A \in \confset_A \mid \exists\: \pid \in \indexset: \statemap\of\pid \in \stateset_F^\program}$, i.e. the set of all configurations where at least one process is in a final state.
The transition relation $\to$ is defined by the following rules:
\begin{itemize}
    \item For each transition $\conf \to[\instr_\pid]_\program \conf'$ where $\conf, \conf' \in \confset^\program$, $\pid \in \indexset$ and $\instr \in \instrs$, it holds that $\conf_A \to \conf'_B$ and $\conf_B \to \conf'_A$.
    This means that each player can execute any TSO instruction, but they take turns alternatingly.
    \item \emph{If player A can update before her own turn:}
    For each transition $\conf_A \to \conf'_B$ introduced by any of the previous rules, it holds that $\tilde\conf_A \to \conf'_B$ for all $\tilde\conf$ with $\tilde\conf \to[\up\kstar]_\program \conf$.
    \item \emph{If player A can update after her own turn:}
    For each transition $\conf_A \to \conf'_B$ introduced by any of the previous rules, it holds that $\conf_A \to \tilde\conf'_B$ for all $\tilde\conf'$ with $\conf' \to[\up\kstar]_\program \tilde\conf'$.
    \item The update rules for player B are defined in a similar manner.
\end{itemize}

From this definition, we obtain 16 different variants of TSO games, which differ in whether each of the players can update \emph{never}, \emph{before} her turn, \emph{after} her turn, or \emph{always} (before and after her turn).
We group games with similar decidability and complexity results together.
An overview of these four groups is presented in \autoref{fig:tso-groups}.
Each group is described in detail in the following sections.

But first, we present a general result that gives a lower complexity bound for all groups of TSO games.
Unexpectedly, even a single process is enough to show \exptime-hardness.
We prove this by reducing the \emph{word acceptance problem} of \emph{linearly bounded alternating Turing machines} (ATM) to the safety problem of a single-process TSO game.
The idea is to store the state and head position of the ATM in the local state of the process, and use a set of variables to save the word on the working tape.
Based on the alternations of the Turing machine, either player A or player B decides which transition the program will simulate.
Interestingly, we can argue that the exact type of TSO game is irrelevant.
Moreover, the construction does not make use of the memory buffers, which implies that the result would even hold if the program followed SC semantics.
The formal proof can be found in Appendix A of the extended version of this paper \cite{spengler2023tso}.

\begin{theorem}
\label{thm:complexity}
    The safety problem for TSO games is \exptime-hard.
\end{theorem}

\section{Group I}

All TSO games in this group have the following in common:
There is one player that can update messages \emph{after} her turn, and the other player can update messages \emph{before} her turn.
Both players might be allowed to do more than that, but fortunately we do not need to differentiate between those cases.
In the following, we call the player that updates after her turn \emph{player X}, and the other one \emph{player Y}.
Although the definition of safety games seems to be of asymmetric nature (player B tries to \emph{reach} a final configuration, while player A tries to \emph{avoid} them), the proof does not rely on the exact identity of player X and Y.

In this section, given a configuration $\conf$, we write $\bar\conf$ to denote the unique configuration obtained from $\conf$ after updating all messages to the memory.
More formally, $\conf \to[\up\kstar] \bar\conf$ and all buffers of $\bar\conf$ are empty.

Let $\game = \tuple{ \confset, \confset_A, \confset_B, \to, \confset_F }$ be a TSO game as described above, currently in some configuration $\conf_0 \in \confset$.
We first consider the situation where player X has a winning strategy $\sigma_X$ from $\conf_0$.
Let $\sigma_Y$ be an arbitrary strategy for player Y and define two more strategies $\bar\sigma_X: \conf \mapsto \overline{\sigma_X\of\conf}$ and $\bar\sigma_Y: \conf \mapsto \sigma_Y(\bar\conf)$.
That is, they act like $\sigma_X$ and $\sigma_Y$, respectively, with the addition that $\bar\sigma_X$ empties the buffer \emph{after} each turn and $\bar\sigma_Y$ empties the buffer \emph{before} each turn.
From the definitions it follows directly that $\bar\sigma_Y(\sigma_X(\conf)) = \sigma_Y(\bar\sigma_X(\conf))$ for all $\conf \in \confset_X$.
An example can be seen in \autoref{fig:group-I}.

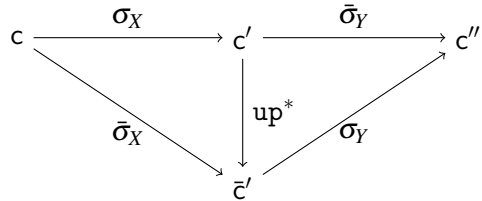
\begin{figure}
\centering
\begin{tikzpicture}[xscale=3,yscale=-2]
    \node	at (0,0)	(c0) {$\conf$};
    \node	at (1,0)	(c1) {$\conf'$};
    \node	at (1,1)	(c1b){$\bar\conf'$};
    \node	at (2,0)	(c2) {$\conf''$};

    \draw[->] (c0) -- node[above] {$\sigma_X$}      (c1);
    \draw[->] (c0) -- node[below] {$\bar\sigma_X$}  (c1b);
    \draw[->] (c1) -- node[right] {$\up\kstar$}     (c1b);
    \draw[->] (c1) -- node[above] {$\bar\sigma_Y$}  (c2);
    \draw[->] (c1b)-- node[below] {$\sigma_Y$}      (c2);
\end{tikzpicture}
\caption{Commutative diagram of strategies in games of group I.}
\label{fig:group-I}
\end{figure}

We argue that $\bar\sigma_X$ is a winning strategy for player X.
The intuition behind this is as follows:
Using the notation of \autoref{fig:group-I}, if a configuration $\conf''$ is reachable from $\bar\conf'$, then it is also reachable from $\conf'$, since player Y can empty all buffers at the start of her turn and then proceed as if she started in $\bar\conf'$.
On the other hand, there might be configurations reachable from $\conf'$ but not $\bar\conf'$, for example a read transition might get disabled by one of the buffer updates.
Thus, player X never gets a disadvantage by emptying the buffers.


\begin{claim}
\label{claim:ab1}
    $\bar\sigma_X$ is a winning strategy from $\conf_0$.
\end{claim}
\begin{proof}
    \underline{\textbf{Case} $\conf_0 \in \confset_X$:}
    Since $\bar\sigma_Y(\sigma_X(\conf)) = \sigma_Y(\bar\sigma_X(\conf))$ for all $\conf \in \confset_X$, the plays $\play_1 = \play(\conf_0, \sigma_X, \bar\sigma_Y)$ and $\play_2 = \play(\conf_0, \bar\sigma_X, \sigma_Y)$ agree on every second configuration, i.e. the configurations in $\confset_X$.
    Moreover, the configurations in between (after an odd number of steps) at least share the same global state, i.e. $\statemap(\sigma_X(\conf)) = \statemap(\bar\sigma_X(\conf))$.
    In particular, the sequence of visited global TSO states is the same in both plays.
    Since $\sigma_X$ is a winning strategy from $\conf_0$, it means that $\play_1$ is winning for player X.
    This means that $\play_2$ is also winning, because for both players, a winning play is clearly determined by the sequence of visited global TSO states.
    Because we chose $\sigma_Y$ arbitrarily, it follows that $\bar\sigma_X$ is a winning strategy.

    \underline{\textbf{Case} $\conf_0 \in \confset_Y$:}
    For the other case, we consider the configurations in $\post(\conf_0)$ instead.
    We observe that $\sigma_X$ must be a winning strategy for all $\conf \in \post(\conf_0)$.
    We apply the first case of this proof to each of these configurations and obtain that $\bar\sigma_X$ is a winning strategy for all of them.
    It follows that $\bar\sigma_X$ is a winning strategy for $\conf_0$.
\end{proof}

Suppose that player X plays her modified strategy as described above.
We observe that after at most two steps, every play induced by her strategy and an arbitrary strategy of the opposing player only visits configurations with at most one message in the buffers:
Player X will empty all buffers at the end of each of her turns and player Y can only add at most one message to the buffers in between.
Hence, they can play on a finite set of configurations instead.

To show this, we construct a finite game $\game' = \tuple{ \confset', \confset_A', \confset_B', \to', \confset_F'}$ as follows.
$\confset_Y'$ contains all configurations of $\confset_Y$ that have at most one buffer message, i.e. $\set{ \tuple{\statemap, \buffermap, \memorymap}_Y \in \confset_Y \mid \sum_{\pid\in\indexset} \sizeof{\buffermap\of\pid} \leq 1 }$.
If $\conf_0 \in \confset_Y$, we also add it to $\confset_Y'$, otherwise to $\confset_X'$.
Lastly, we add $\post(\confset_Y')$ to $\confset_X'$, where $\post$ is with respect to $\game$.
$\to'$ is defined as the restriction of $\to$ to configurations of $\game'$, and $\confset_F' = \confset_F \cap \confset_A'$.
Note that $\confset_X'$ also contains configurations with two messages.
This is needed to account for the case that player Y has a winning strategy, which is handled later in this proof.
Now, let $\bar\sigma_X'$ be the restriction of $\bar\sigma_X$ to $\confset_X'$ (in the mathematical sense, i.e $\bar\sigma_X': \confset_X' \to \confset_Y$ and $\bar\sigma_X\of\conf = \bar\sigma_X'\of\conf$ for all $\conf \in \confset_X'$).

\begin{claim}
\label{claim:ab2}
    $\bar\sigma_X'$ is a winning strategy for $\conf_0$ in $\game'$.
\end{claim}
\begin{proof}
    Looking at the definitions, we confirm that $\bar\sigma_X'$ actually is a valid strategy for $\game'$, i.e. $\bar\sigma_X'(\conf) \in \confset_Y'$, for all $\conf \in \confset_X'$, since $\bar\sigma_X'(\conf)$ has empty buffers.
    (This makes $\bar\sigma_X'$ the restriction of $\bar\sigma_X$ to $\game'$.)
    Consider a strategy $\sigma_Y'$ for player Y in $\game'$ and an arbitrary extension $\sigma_Y$ to $\game$.
    Because $\bar\sigma_X'$ and $\bar\sigma_X$ agree on $\confset_X'$ and $\bar\sigma_Y'$ and $\bar\sigma_Y$ agree on $\confset_Y'$, $\play = \play(\conf_0, \bar\sigma_X', \bar\sigma_Y)$ and $\play' = \play(\conf_0, \bar\sigma_X', \bar\sigma_Y)$ are in fact the exact same play.
    Since $\bar\sigma_X$ is a winning strategy, $\play$ is a winning play, and thus also $\play'$.
    Here, note that $\game$ and $\game'$ agree on the final configurations within $\confset'$.
    Since $\sigma_Y'$ was arbitrary, it follows that $\bar\sigma_X'$ is a winning strategy from $\conf_0$ in $\game'$.
\end{proof}

What is left to show is that a winning strategy for $\game'$ induces a winning strategy for $\game$.
Suppose $\sigma_X'$ is a winning strategy for player X in game $\game'$ for the configuration $\conf_0$.
Let $\sigma_X$ be an arbitrary extension of $\sigma_X'$ to $\game$.

\begin{claim}
\label{claim:ab3}
    $\sigma_X$ is a winning strategy for $\conf_0$ in $\game$.
\end{claim}
\begin{proof}
    Let $\sigma_Y$ be a strategy of player Y in $\game$ and $\sigma_Y'$ the restriction of $\sigma_Y$ to $\confset_Y'$ (again, in the mathematical sense).
    Since the outgoing transitions of every $\conf \in \confset_Y'$ are the same in both $\game$ and $\game'$, $\sigma_Y'$ is a strategy for $\game'$ (and the restriction of $\sigma_Y$ to $\game'$).
    Furthermore, starting from $\conf_0$, we see that $\sigma_X$ and $\sigma_Y$ induce the exact same play in $\game$ as $\sigma_X'$ and $\sigma_Y'$ in $\game'$.
    Since the former play is winning, so must be the latter one.
\end{proof}

Now, we quickly cover the situation where it is player Y that has a winning strategy.
We follow the same arguments as previously, with minor changes.
This time, assume $\sigma_Y$ to be a winning strategy and let $\sigma_X$ be arbitrary.
Define $\bar\sigma_X$ and $\bar\sigma_Y$ as above.
Following the beginning of the proof of \autoref{claim:ab1}, we can conclude that the sequence of visited global TSO states is the same in both play $\play_1$ and $\play_2$.
For the remainder of the proof, we swap the roles of X and Y and obtain that $\bar\sigma_Y$ is a winning strategy.

Let $\bar\sigma_Y'$ be the restriction of $\bar\sigma_Y$ to $\confset_Y'$.
Since $\bar\sigma_Y'(\confset_Y') = \bar\sigma_Y(\confset_Y') \subseteq \post(\confset_Y') \subseteq \confset_X'$, it follows that $\bar\sigma_Y'$ is a strategy of $\game'$ ($\post$ is again with respect to $\game$).
Consider a strategy $\sigma_X'$ for player X in $\game'$ and an arbitrary extension $\sigma_X$ to $\game$.
Similar as in \autoref{claim:ab2}, we see that $\play(\conf_0, \bar\sigma_X', \bar\sigma_Y) = \play(\conf_0, \bar\sigma_X', \bar\sigma_Y)$ and conclude that $\bar\sigma_Y'$ is a winning strategy.

The other direction follows from the proof of \autoref{claim:ab3}, with the roles of X and Y swapped.

\begin{theorem}
\label{thm:ab}
    The safety problem for games of group I is \exptime-complete.
\end{theorem}
\begin{proof}
    By \autoref{claim:ab1} and \autoref{claim:ab2}, if a configuration $\conf_0$ is winning for player X in $\game$, then it is also winning in $\game'$.
    The reverse holds by \autoref{claim:ab3}.
    The equivalent statement for player Y follows from results outlined above.
    Thus, the safety problem for $\game$ is equivalent to the safety problem for $\game'$.
    $\game'$ is finite and has exponentially many configurations.
    \exptime-completeness follows immediately from \autoref{lem:finite} (membership) and \autoref{thm:complexity} (hardness).
\end{proof}

\begin{remark}
    In the game where both players are allowed to update the buffer at any time, we can show an interesting conclusion.
    By \autoref{claim:ab1} and the equivalent statement for the second player, we can restrict both players to strategies that empty the buffer after each turn.
    Thus, the game is played only on configurations with empty buffer, except for the initial configuration which might contain some buffer messages.
    This implies that the TSO program that is described by the game implicitly follows SC semantics.
\end{remark}

\section{Group II}

This group contains TSO games where both players are allowed to update the buffer \emph{only} before their own move.
Let player X be the player that has a winning strategy and player Y her opponent.
Note that this differs from the previous section, in which the players X and Y were defined based on their updating capabilties.

Similar to the argumentation for Group I, we want to show that player X also has a winning strategy where she empties the buffer in each move.
But, in contrast to before, this time there is an exception:
Since the player has to update the buffer \emph{before} her move, by updating a memory variable she might disable a read transition that she intended to execute.
Thus, we do not require her to empty the buffer in that case.

Formally, let $\game = \tuple{ \confset, \confset_X, \confset_Y, \to, \confset_F}$ be a TSO game where both players are allowed to perform buffer updates exactly before their own moves.
Suppose $\sigma_X$ is a winning strategy for player X and some configuration $\conf_0$.
We construct another strategy $\bar\sigma_X$ for player X.
Let $\conf \in \confset_X$, $\conf' = \sigma_X(\conf)$ and $\bar\conf$ as in the previous section, i.e. the unique configuration such that $\conf \to[\up\kstar]_\program \bar\conf$ and the buffers of $\conf$ are empty.
Suppose that $\conf \to[\instr_\pid]_\program \conf'$, where $\instr_\pid$ is not a read instruction.
Then, starting from $\conf$, updating all buffer messages does not change that the transition from $\statemap(\conf)\of\pid$ to $\statemap(\conf')\of\pid$ is enabled.
Thus, $\instr_\pid$ can also be executed from $\bar\conf$.
We call the resulting configuration $\tilde\conf'$ and observe that $\bar\conf \to_\program \tilde\conf'$ and $\conf' \to[\up\kstar] \tilde\conf'$.
We define $\bar\sigma_X(\conf) = \tilde\conf'$.
This can be seen in \autoref{fig:group-II}.
Note that $\tilde\conf'$ may have at most one message in its buffers.
In the other case, where there is no transition from $\conf$ to $\conf'$ other than read instructions, we define $\bar\sigma_X(\conf) = \sigma_X(\conf) = \conf'$.

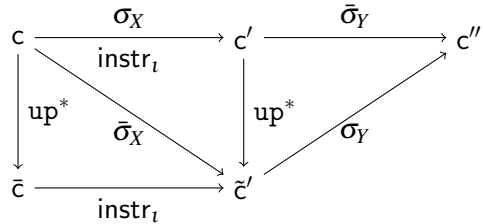
\begin{figure}
\centering
\begin{tikzpicture}[xscale=3,yscale=-1]
    \node	at (0,-1)	(c0) {$\conf$};
    \node	at (0,1)	(c0b){$\bar\conf$};
    \node	at (1,-1)	(c1) {$\conf'$};
    \node	at (1,1)	(c1t){$\tilde\conf'$};
    \node	at (2,-1)	(c2) {$\conf''$};

    \draw[->] (c0) -- node[above] {$\sigma_X$}
                      node[below] {$\instr_\pid$}   (c1);
    \draw[->] (c0) -- node[below] {$\bar\sigma_X$}  (c1t);
    \draw[->] (c0) -- node[right] {$\up\kstar$}     (c0b);
    \draw[->] (c0b)-- node[below] {$\instr_\pid$}   (c1t);
    \draw[->] (c1) -- node[right] {$\up\kstar$}     (c1t);
    \draw[->] (c1) -- node[above] {$\bar\sigma_Y$}  (c2);
    \draw[->] (c1t)-- node[below] {$\sigma_Y$}      (c2);

\end{tikzpicture}
\caption{Commutative diagram of strategies in games of group II, in the case where $\instr_\pid \neq \rd\of\xd$.}
\label{fig:group-II}
\end{figure}

\begin{claim}
\label{claim:bb1}
    $\bar\sigma_X$ is a winning strategy for $\conf_0$.
\end{claim}
\begin{proof}
    First, suppose that $\conf_0 \in \confset_X$ and let $\sigma_Y$ be an arbitrary strategy of player Y.
    We define another (non-positional) strategy $\bar\sigma_Y$, that depends on the last two configurations, by $\bar\sigma_Y(\conf, \conf') = \sigma_Y(\bar\sigma_X(\conf))$.
    We observe that for all $\conf \in \confset_X$, it holds that $\bar\sigma_Y(\conf, \sigma_X(\conf)) = \sigma_Y(\bar\sigma_X(\conf))$.
    It follows that the play $\play_1$ induced by $\sigma_X$ and $\bar\sigma_Y$ and the play $\play_2$ induced by $\bar\sigma_X$ and $\sigma_Y$ agree on every second configuration, i.e. the configurations in $\confset_X$.
    In particular, the sequence of visited global TSO configurations is the same in both plays.
    Since $\sigma_X$ is winning, it means that $\play_1$ is winning for player X and thus also $\play_2$ is winning.
    Because we chose $\sigma_Y$ arbitrarily, it follows that $\bar\sigma_X$ is a winning strategy.

    Otherwise, if $\conf_0 \in \confset_Y$, we consider the successors of $\conf_0$ instead.
    We note that $\bar\sigma_X$ must also be a winning strategy for each $\conf \in \post(\conf_0)$.
    But then, we can apply the previous arguments to each of those configurations and conclude that $\bar\sigma_X$ is a winning strategy for all of them.
    Thus, it is also a winning strategy for $\conf_0$.
\end{proof}

We conclude that if player X has a winning strategy $\sigma_X$, then she also has a winning strategy $\bar\sigma_X$ where she empties the buffers before every turn in which she does not perform a read operation.
By symmetry, the same holds true for player Y.
Thus, we can limit our analysis to this type of strategies.
We see that the number of messages in the buffers is bounded:
Suppose that the game is in configuration $\conf \in \confset_X$.
Then, $\bar\sigma_X$ either empties the buffer and adds at most one new message, or it performs a transition due to a read instruction, which does not increase the size of the buffers.
The analogous argumentation holds for player Y.
Hence, we can reduce the game to a game on bounded buffers, which is finite state and thus decidable.

Given the configuration $\conf_0$ as above, we construct a finite game $\game' = \tuple{ \confset', \confset_X', \confset_Y', \to', \confset_F'}$ as follows.
The set $\confset_X'$ contains all configurations from $\confset_X$ which have at most as many buffer messages than $\conf_0$ (or at most one message, if $\conf_0$ has empty buffers):
$\confset_X' = \Set{ \conf \in \confset_X \mid \sizeof{\buffermap\of\conf} \leq \max\set{1, \sizeof{\buffermap\of{\conf_0}}} }$, where $\sizeof\buffermap = \sum_{\pid\in\indexset} \sizeof{\buffermap\of\pid}$.
The set $\confset_Y'$ is defined accordingly.
Note that both sets are finite.
Lastly, $\to'$ is defined as the restriction of $\to$ to configurations of $\game'$, and $\confset_F' = \confset_F \cap \confset_A'$.
We define $\bar\sigma_X'$ to be the restriction of $\bar\sigma_X$ to $\confset_X'$.
Since $\bar\sigma_X'(\conf) \in \confset_Y'$ for all $\conf \in \confset_X'$, $\bar\sigma_X'$ is indeed a valid strategy for $\game'$.
In particular, it is the restriction of $\bar\sigma_X$ to $\game'$.

\begin{claim}
\label{claim:bb2}
    $\bar\sigma_X'$ is a winning strategy for $\conf_0$ in $\game'$.
\end{claim}
\begin{proof}
    First, consider the case where $\conf_0 \in \confset_X$.
    Let $\sigma_Y'$ be a strategy for player Y in $\game'$ and let $\sigma_Y$ be an arbitrary extension of $\sigma_Y'$ to $\game$.
    The play $\play$ induced by $\bar\sigma_X$ and $\sigma_Y$ in $\game$ is the same as the play $\play'$ induced by $\bar\sigma_X'$ and $\sigma_Y'$ in $\game'$.
    Since $\bar\sigma_X$ is a winning strategy, $\play$ is a winning play.
    It follows that $\play'$ must also be a winning strategy.
    Since $\sigma_Y'$ was arbitrary, it follows that $\bar\sigma_X'$ is a winning strategy and $\conf_0$ is winning in $\game'$.
\end{proof}

\begin{theorem}
    The safety problem for games of group II is \exptime-complete.
\end{theorem}
\begin{proof}
    By \autoref{claim:bb1} and \autoref{claim:bb2}, if a configuration $\conf_0$ is winning for player A in game $\game$, then it is also winning in $\game'$.
    The same holds true for player B.
    Thus, the safety problem for $\game$ is equivalent to the safety problem for $\game'$.
    Similar to the games of group I, $\game'$ is finite and has exponentially many configurations.
    By \autoref{lem:finite} and \autoref{thm:complexity}, we can again conclude that the safety problem is \exptime-complete.
\end{proof}

\section{Group III}
\label{sec:group-III}

This group consists of all games where exactly one player has control over the buffer updates, and additionally the game where both players are allowed to update buffer messages \emph{after} their own move.
Intuitively, all of them have in common that the TSO program can attribute a buffer update to one specific player.
If only one player can update messages, this is clear.
In the other game, the first player who observes that a buffer message has reached the memory is not the one who has performed the buffer update.
Thus, the program is able to punish misbehaviour, i.e. not following protocols or losing messages.

We will show that the safety problem is undecidable for this group of games.
To accomplish that, we reduce the state reachability problem of PCS to the safety problem of each game.
Since the former problem is undecidable, so is the latter.

The case where player A is allowed to perform buffer updates at any time is called the \emph{A-TSO game}.
It is explained in detail in the following.
The other cases work similar, but require slightly different program constructions.
They are presented in the appendix \cite{spengler2023tso}.

\medskip

Consider the A-TSO game, i.e. the case where player A can update messages at any time, but player B can never do so.
Given a PCS $\channelsystem = \tuple{ \channelstateset, \channelmessageset, \to_\channelsystem }$ and a final state $\channelstate_F \in \channelstateset$, we construct a TSO program $\program$ that simulates $\channelsystem$.
We design the program such that $\channelstate_F$ is reachable in $\channelsystem$ if and only if player B wins the safety game induced by $\program$.
Thus, the construction gives her the initiative to decide which transitions of $\channelsystem$ will be simulated.
Meanwhile, the task of player A is to take care of the buffer updates.

$\program$ consists of three processes $\process^1$, $\process^2$ and $\process^3$, that operate on the variables $\set{ \xwr, \xrd, \yvar }$ over the domain $\channelmessageset \uplus \set{ 0, 1, \bot }$.
The first process simulates the control flow and the message channel of the PCS $\channelsystem$.
The second process provides a mean to read from the channel.
The only task of the third process is to prevent deadlocks, or rather to make any deadlocked player lose.
$\process^3$ achieves this with four states: the initial state, an intermediate state, and one winning state for each player, respectively.
If one of the players cannot move in both $\process^1$ and $\process^2$, they have to take a transition in $\process^3$.
From the initial state of this process, there exists only one outgoing transition, which is to the intermediate state.
From there, the other player can move to her respective winning state and the process will only self-loop from then on.
For player A, her state is winning because she can refuse to update any messages, which will ensure that player B keeps being deadlocked in $\process^1$ and $\process^2$.
For player B, her state simply is contained in $\stateset_F^\program$.
In the following, we will mostly omit $\process^3$ from the analysis and just assume that both players avoid reaching a configuration where they cannot take any transition in either $\process^1$ or $\process^2$.

As mentioned above, we will construct $\process^1$ and $\process^2$ to simulate the perfect channel system in a way that gives player B the control about which channel operation will be simulated.
To achieve this, each channel operation will need an even number of transitions to be simulated in $\program$.
Since player B starts the game, this means that after every fully completed simulation step, it is again her turn and she can initiate another simulation step as she pleases.
Furthermore, during the simulation of a skip or send operation, we want to prevent player A from executing $\process^2$, since this process is only needed for the receive operation.
Suppose that we want to block player A from taking a transition $\state \to[\instr]_\program \state'$.
We add a new transition $\state' \to[\nop]_\program \state_F$, where $\state_F \in \statemap_F^\program$.
Hence, reaching $\state'$ is immediately losing for player A, since player B can respond by moving to $\state_F$.

Next, we will describe how $\process^1$ and $\process^2$ simulate the perfect channel system $\channelsystem$.
For each transition in $\channelsystem$, we construct a sequence of transitions in $\process^1$ that simulates both the state change and the channel behaviour of the $\channelsystem$-transition.
To achieve this, $\process^1$ uses its buffer to store the messages of the PCS's channel.
In particular, to simulate a send operation $!\channelmessage$, $\process^1$ adds the message $\tuple{\xwr, \channelmessage}$ to its buffer.
For receive operations, $\process^1$ cannot read its own oldest buffer message, since it is overshadowed by the more recent messages.
Thus, the program uses $\process^2$ to read the message from memory and copies it to the variable $\xrd$, where it can be read by $\process^1$.
We call the combination of reading a message $\channelmessage$ from $\xwr$ and writing it to $\xrd$ the \emph{rotation} of $\channelmessage$.

While this is sufficient to simulate all behaviours of the PCS, it also allows for additional behaviour that is not captured by $\channelsystem$.
More precisely, we need to ensure that each channel message is received \emph{once and only once}.
Equivalently, we need to prevent the \emph{loss} and \emph{duplication} of messages.
This can happen due to multiple reasons.

The first phenomenon that allows the loss of messages is the seeming lossiness of the TSO buffer.
Although it is not strictly lossy, it can appear so:
Consider an execution of $\program$ that simulates two send operations $!\channelmessage_1$ and $!\channelmessage_2$, i.e. $\process^1$ adds $\tuple{\xwr, \channelmessage_1}$ and $\tuple{\xwr, \channelmessage_2}$ to its buffer.
Assume that player A decides to update both messages to the memory, without $\process^2$ performing a message rotation in between.
The first message $\channelmessage_1$ is overwritten by the second message $\channelmessage_2$ and is lost beyond recovery.

To prevent this, we extend the construction of $\process^1$ such that it inserts an auxiliary message $\tuple{\yvar, 1}$ into its buffer after the simulation of each send operation.
After a message rotation, that is, after $\process^2$ copied a message from $\xwr$ to $\xrd$, the process then resets the value of $\xwr$ to its initial value $\bot$.
Next, the process checks that $\yvar$ contains the value $0$, which indicates that only one message was updated to the memory.
Now, player A is allowed to update exactly one $\tuple{\yvar, 1}$ buffer message, after which $\process^2$ resets $\yvar$ to $0$.
To ensure that player A has actually updated only one message in this step, $\process^2$ then checks that $\xwr$ is still empty.
Since player A is exclusively responsible for buffer updates, $\process^2$ deadlocks her whenever one of these checks fails.

In the next scenario, we discover a different way of message loss.
Consider again an execution of $\program$ that simulates two send operations $!\channelmessage_1$ and $!\channelmessage_2$.
Assume Player A updates $\channelmessage_1$ to the memory and $\process^2$ performs a message rotation.
Immediately afterwards, the same happens to $\channelmessage_2$, without $\process^1$ simulating a receive operation in between.
Again, $\channelmessage_1$ is overwritten by $\channelmessage_2$ before being received, thus it is lost.

Player A is prevented from losing a message in this way by disallowing her to perform a complete message rotation (including the update of one $\tuple{\yvar,1}$-message and the reset of the variables) entirely on her own.
More precisely, we add a winning transition for player B to $\process^2$ that she can take if and only if player A is the one initiating the update of $\tuple{\yvar,1}$.
On the other hand, player A can prevent player B from performing two rotations right after each other by refusing to update the next buffer message until $\process^1$ initiates the simulation of a receive operation.

Lastly, we investigate message duplication.
This occurs if $\process^1$ simulates two receive operations without $\process^2$ performing a message rotation in between.
In this case, the most recently rotated message is received twice.

The program prevents this by blocking $\process^1$ from progressing after a receive operation until $\process^2$ has finished a full rotation.
In detail, at the very end of the message rotation and $\tuple{\yvar,1}$-update, $\process^2$ reset the value of $\xrd$ to its initial value $\bot$.
After simulating a receive operation, $\process^1$ is blocked until it can read this value from memory.

This concludes the mechanisms implemented to ensure that each channel message is received \emph{once and only once}.
Thus, we have constructed an A-TSO game that simulates a perfect channel system.
We summarise our results in the following theorem.
The formal proof can be found in Appendix B \cite{spengler2023tso}.

\begin{theorem}
\label{thm:atso}
    The safety problem for the A-TSO game is undecidable.
\end{theorem}

\section{Group IV}
\label{sec:group-IV}

In TSO games where no player is allowed to perform any buffer updates, there is no communication between the processes at all.
A read operation of a process $\process^\pid$ on a variable $\xvar$ either reads the initial value from the shared memory, or the value of the last write of $\process^\pid$ on $\xvar$ from the buffer, if such a write operation has happened.

Thus, we are only interested in the transitions that are enabled for each process, but we do not need to care about the actual buffer content.
In particular, the information that we need to capture from the buffers and the memory is the values that each process can read from the variables, and whether a process can execute a memory fence instruction or not.
Together with the global state of the current configuration, this completely determines the enabled transitions in the system.

We call this concept the \emph{view} of the processes on the concurrent system and define it formally as a tuple $\view = \tuple{ \statemap, \valuemap, \fencemap }$, where:
\begin{itemize}
    \item $\statemap: \indexset \to \bigcup_{\pid \in \indexset} \stateset^\pid$ is a global state of $\program$.
    \item $\valuemap: \indexset \times \varset \to \valset$ defines which value each process reads from a variable.
    \item $\fencemap: \indexset \to \set{ \true, \false }$ represents the possibility to perform a memory fence instruction.
\end{itemize}
Given a view $\view = \tuple{ \statemap, \valuemap, \fencemap }$, we write $\statemap\of\view$, $\valuemap\of\view$ and $\fencemap\of\view$ for the global program state $\statemap$, the value state $\valuemap$ and the fence state $\fencemap$ of $\view$.

The view of a configuration $\conf$ is denoted by $\view\of\conf$ and defined in the following way.
First, $\statemap(\view\of\conf) = \statemap\of\conf$.
For all $\pid \in \indexset$ and $\xvar \in \varset$, if $\buffermap\of\conf\of\pid|_{\set\xvar\times\valset} = \tuple\xd \cdot \word$, then $\valuemap(\view\of\conf)(\pid, \xvar) = \dval$.
Otherwise, $\valuemap(\view\of\conf)(\pid, \xvar) = \memorymap\of\conf\of\xvar$.
Lastly, $\fencemap(\view\of\conf)\of\pid = \true$ if and only if $\buffermap\of\conf\of\pid = \varepsilon$.
We extend the notation to sets of configurations in the usual way, i.e. $\view(\confset') = \set{ \view(\conf) \mid \conf \in \confset'}$.

For $\conf, \conf' \in \confset_\program$, if $\view(\conf) = \view(\conf')$, then we write $\conf \equiv \conf'$ and say that $\conf$ and $\conf'$ are \emph{view-equivalent}.
In such a case, a local process of $\program$ cannot differentiate between $\conf$ and $\conf'$ in the sense that the enabled transitions in both configurations are the same.
\autoref{lem:views} captures this idea formally.

\begin{lemma}
\label{lem:views}
    For all $\conf_1, \conf_2, \conf_3 \in \confset_\program$, $\pid \in \indexset$ and $\instr \in \instrs$ with $\conf_1 \to[\instr_\pid] \conf_2$ and $\conf_1 \equiv \conf_3$,
    there exists a $\conf_4 \in \confset_\program$ such that $\conf_3 \to[\instr_\pid] \conf_4$ and $\conf_2 \equiv \conf_4$.
\end{lemma}
\begin{proof}
    We first show that $\instr_\pid$ is enabled at $\conf_3$.
    Since $\conf_1 \equiv \conf_3$, it holds that $\statemap(\conf_1) = \statemap(\conf_3)$.
    Furthermore, if $\instr_\pid = \rd\of\xd_\pid$, then $\valuemap(\view(\conf_1))(\pid, \xvar) = \valuemap(\view(\conf_3))(\pid, \xvar) = \dval$.
    Also, if $\instr_\pid = \mf_\pid$, then $\fencemap(\view(\conf_1))(\pid) = \fencemap(\view(\conf_3))(\pid) = \varepsilon$.
    From these considerations and the definition of the TSO semantics (see \autoref{fig:tso-semantics}), it follows that $\instr_\pid$ is indeed enabled at $\conf_3$.

    Let $\conf_4$ be the configuration obtained after performing $\instr_\pid$, i.e. $\conf_3 \to[\rd\of\xd_\pid] \conf_4$.
    It holds that $\statemap(\conf_4) = \statemap(\conf_2) = \statemap(\conf_1)[\pid \leftarrow \statemap(\conf_2)\of\pid]$.
    If $\instr_\pid = \wr\of\xd_\pid$, then $\valuemap(\view(\conf_4)) = \valuemap(\view(\conf_2)) = \valuemap(\view(\conf_1))[(\pid, \xvar) \leftarrow \dval]$
    and $\fencemap(\view(\conf_4)) = \fencemap(\view(\conf_2)) = \fencemap(\view(\conf_1))[\pid \leftarrow \false]$.
    Otherwise, $\valuemap(\view(\conf_4)) = \valuemap(\view(\conf_2)) = \valuemap(\view(\conf_1))$
    and $\fencemap(\view(\conf_4)) = \fencemap(\view(\conf_2)) = \fencemap(\view(\conf_1))$.
    In all cases it follows that $\conf_2 \equiv \conf_4$.
\end{proof}

We define a finite safety game played on TSO views and show that we can restrict our analysis to this game.
Let $\game = \tuple{ \confset, \confset_A, \confset_B, \to, \confset_F}$ be a TSO game where neither player can perform any updates.
We define a new game $\game' = \tuple{ \viewset, \viewset_A, \viewset_B, \to', \viewset_F}$ that is played on the views of $\game$.
We define $\viewset_A = \set{ \view\of\conf_A \mid \conf_A \in \confset_A }$, $\viewset_B = \set{ \view\of\conf_B \mid \conf_B \in \confset_B }$, $\viewset = \viewset_A \cup \viewset_B$ and $\viewset_F = \set{ \view\of\conf_A \mid \conf_A \in \confset_F }$.
Lastly, $\view(\conf) \to' \view(\conf')$ if and only if $\conf \to \conf'$.
This is well-defined by \autoref{lem:views}.

\begin{lemma}
\label{lem:no-updates}
    A configuration $\conf_0 \in \confset$ is winning (for player A / B) in $\game$ if and only if the view $\view_0 = \view(\conf_0) \in \viewset$ is winning (for player A / B) in $\game'$.
\end{lemma}
\begin{proof}
    To simplify notation, we extend $\view\of\conf$ to configurations of TSO games by $\view(\conf_A) = \view(\conf)_A$ and $\view(\conf_B) = \view(\conf)_B$ for $\conf_A \in \confset_A$ and $\conf_B \in \confset_B$.
    Hence, we can write $\viewset_A = \view(\confset_A)$ and similar.

    Suppose $\conf_0$ is winning for some player X with (positional) strategy $\sigma_X$ and consider the case $\conf_0 \in \confset_X$.
    In the following, we will define a (non-positional) strategy $\sigma_X'$ for $\game'$.

    First, we need an auxiliary function $f: \confset \times \viewset \to \confset$ that fulfills the condition:
    For all $\conf \in \confset$ and $\view \in \viewset$ such that $\view(\conf) \to' \view$, it holds that $\conf \to f(\conf, \view)$ and $\view = \view(f(\conf, \view))$.
    Intuitively, $f$ selects a successor of $\conf$ with view $\view$.
    Such a function exists by \autoref{lem:views}.

    For $n$ even and a sequence $\view_0, \dots, \view_n$, iteratively define $\conf_{2i-1} = \sigma(\conf_{2i-2})$ and $\conf_{2i} = f(\conf_{2i-1}, \view_{2i})$ for $i = 1, \dots, n/2$.
    Then, $\sigma_X'(\view_0, \dots, \view_n) = \view(\sigma_X(\conf_n))$.
    We will show that $\sigma_X'$ is a winning strategy for $\view_0$.
    Consider a positional strategy $\sigma_Y'$ for player Y in $\game'$.
    We define a positional strategy $\sigma_Y$ for player Y in $\game$ by $\sigma_Y(\conf) = f(\conf, \sigma_Y'(\view\of\conf))$.
    Consider the play $\play = \conf_0, \conf_1, \dots$ induced by $\sigma_X$ and $\sigma_Y$, and the play $\play' = \view_0, \view_1, \dots$ induced by $\sigma_X'$ and $\sigma_Y'$.

    We proof by induction over $k$, that (i) $\view_k = \view(\conf_k)$ and (ii) $\conf_k$ of $\play$ coincides with $\conf_k$ as in the definition of $\sigma_X'$.
    In this context, we refer to the latter with $\bar\conf_k$.
    For $k=0$, $\view_0 = \view(\conf_0)$ and $\conf_0 = \bar\conf_0$ by definition.
    For $k$ odd, $\conf_k = \sigma_X(\conf_{k-1}) = \bar\conf_k$ by the induction hypothesis.
    Also,
    $$\view_k = \sigma_X'(\view_0, \dots, \view_{k-1}) = \view(\sigma_X(\bar\conf_{k-1})) = \view(\sigma_X(\conf_{k-1})) = \view(\conf_k) \ .$$
    For $k>0$ even,
    $$\conf_k = \sigma_Y(\conf_{k-1}) = f(\conf_{k-1}, \sigma_Y'(\view(\conf_{k-1}))) = f(\bar\conf_{k-1}, \sigma_Y'(\view_{k-1})) = f(\bar\conf_{k-1}, \view_k) = \bar\conf_k \ .$$
    Lastly,
    $$\view(\conf_k) = \view(\sigma_Y(\conf_{k-1})) = \view(f(\conf_{k-1}, \sigma_Y'(\view(\conf_{k-1})))) = \view(f(\conf_{k-1}, \sigma_Y'(\view_{k-1}))) = \view(f(\conf_{k-1}, \view_k)) = \view_k \ ,$$
    where the last equality follows from the definition of $f$.

    Since $\sigma_X$ is a winning strategy for $\conf_0$, $\play$ is a winning play for player X.
    From the definition of $\viewset_F$ it follows that $\play'$ is a winning play in $\game'$ and thus $\view_0$ is winning for player X.
    Note that by \autoref{lem:positional}, we could have chosen a positional strategy in place of $\sigma_X'$.
    Since we did not put any restrictions on the identity of player X, this concludes both the \emph{if} and the \emph{only if} direction of the proof for the case $\conf_0 \in \confset_X$.

    Otherwise, if $\conf_0 \in \confset_Y$, we consider all configurations of $\post(\conf_0)$ instead.
    We have the following chain of equivalences:
    $\conf_0$ is winning $\iff$ all $\conf \in \post(\conf_0)$ are winning $\iff$ all $\view \in \view(\post(\conf_0))$ are winning $\iff$ all $\view \in \post(\view(\conf_0))$ are winning $\iff $ $\view(\conf_0)$ is winning.
    Here, the second equivalence applies the first case of this proof and the third equivalence uses $\post(\view(\conf_0)) = \view(\post(\conf_0))$, which follows from the definition of $\game'$.
\end{proof}

\begin{theorem}
    The safety problem for games in group IV is \exptime-complete.
\end{theorem}
\begin{proof}
    By \autoref{lem:no-updates}, the safety problem for $\game$ is equivalent to the safety problem of $\game'$, which is played on views.
    Since there exist only exponentially many views, \exptime-completeness follows from \autoref{lem:finite} and \autoref{thm:complexity}, similar to Group I and II.
\end{proof}

\section{Conclusion and Future Work}
In this work we have addressed for the first time the game problem for programs running under weak memory models in general and TSO in particular.
Surprisingly, our results show that depending on when the updates take place, the problem can turn out to be undecidable or decidable.
In fact, there is a subtle difference between the decidable (group I, II and IV) and undecidable (group III) TSO games.
For the former games, when a player is taking a turn, the system does not know who was responsible for the last update.
But for the latter games, the last update can be attributed to a specific player.
Another surprising finding is the complexity of the game problem for the groups I, II and IV which is \exptime-complete in contrast with the non-primitive recursive complexity of the reachability problem for programs running under TSO and the undecidability of the repeated reachability problem.

In future work, the games where exactly one player has control over the buffer seem to be the most natural ones to expand on.
In particular, the A-TSO game (where player A can update before and after her move) and the B-TSO game (same, but for player B).
On the other hand, the games of groups I, II and IV seem to be degenerate cases and therefore rather uninteresting.
In particular, they do not seem to be more powerful than games on programs that follow SC semantics.

Another direction for future work is considering other memory models, such as the partial store ordering semantics, the release-acquire semantics, and the ARM semantics.
It is also interesting to define stochastic games for programs running under TSO as extension of the probabilistic TSO semantics \cite{DBLP:conf/esop/AbdullaAAGK22}.

\newpage
\bibliographystyle{eptcs}
\bibliography{bibdatabase}

\end{document}